\documentclass[conference]{IEEEtran}

\usepackage{graphicx}
\usepackage[font=small]{caption}
\usepackage[labelsep=period]{caption}
\usepackage{float}
\usepackage[noadjust]{cite}
\usepackage{amsfonts}
\usepackage{amsmath}
\usepackage{amsthm}
\usepackage{mathrsfs}
\usepackage{amssymb,mathrsfs}
\usepackage[mathscr]{euscript}
\usepackage{subcaption}
\usepackage{epstopdf} 
\usepackage[dvipsnames]{xcolor}
\usepackage{pgfplots}
\pgfplotsset{compat=newest}
\usetikzlibrary{plotmarks}
\usepackage[normalem]{ulem}
\usepackage{gensymb}
\usepackage{flushend}

\usepackage[colorinlistoftodos]{todonotes}

\usepackage[english]{babel}
\usepackage[utf8]{inputenc}
\usepackage{algorithm}
\usepackage[noend]{algpseudocode}

\usepackage[cmintegrals]{newtxmath}

\pgfplotsset{every axis/.append style={
		scaled x ticks = false, 
		x tick label style={/pgf/number format/.cd, fixed, fixed zerofill,
			int detect,1000 sep={},precision=3}
	}
}

\makeatletter
\def\BState{\State\hskip-\ALG@thistlm}
\makeatother

\definecolor{darkpastelgreen}{rgb}{0.01, 0.75, 0.24}

\DeclareMathOperator*{\argmin}{argmin}

\def\plos{\mathcal{P}_\text{LoS}}

\newtheorem{theorem}{Theorem}

\begin{document}

\title{Energy-Constrained UAV Trajectory Design for Ground Node Localization}

\author{Hazem Sallouha, 
	 Mohammad Mahdi Azari, 
	 and~Sofie~Pollin 
	 \\
\IEEEauthorblockA{Department of Electrical Engineering - TELEMIC, KU Leuven\\
Email: hazem.sallouha@esat.kuleuven.be \\
}}

\maketitle
\begin{abstract}
The use of aerial anchors for localizing terrestrial nodes has recently been recognized as a cost-effective, swift and flexible solution for better localization accuracy, providing localization services when the GPS is jammed or satellite reception is not possible. In this paper, the localization of terrestrial nodes when using mobile unmanned aerial vehicles (UAVs) as aerial anchors is presented. We propose a novel framework to derive localization error in urban areas. In contrast to the existing works, our framework includes height-dependent UAV to ground channel characteristics and a highly detailed UAV energy consumption model. This enables us to explore different trade-offs and optimize UAV trajectory for minimum localization error. In particular, we investigate the impact of UAV altitude, hovering time, number of waypoints and path length through formulating an energy-constrained optimization problem. Our results show that increasing the hovering time decreases the localization error considerably at the cost of a higher energy consumption. To keep the localization error below 100m, shorter hovering is only possible when the path altitude and radius are optimized. For a constant hovering time of 5 seconds, tuning both parameters to their optimal values brings the localization error from 150m down to 65m with a power saving around 25\%.
\end{abstract}
\begin{IEEEkeywords}
Unmanned aerial vehicle (UAV), localization, received signal strength, trajectory planning, optimal altitude
\end{IEEEkeywords}

\IEEEpeerreviewmaketitle

\section{Introduction}

\subsection{Motivation}
Location-aware services have been acknowledged as an indispensable functionality for a vast majority of wireless communication applications. In fact, location information can be exploited in different layers, from communication aided purposes to the application level where location information is needed to meaningfully interpret the collected data \cite{dammann2013where2}. To this end, the global positioning system (GPS) is used for outdoor applications where it provides a satisfactory performance. However, GPS is known of its expensive cost and vulnerability to jamming \cite{alshrafi2014compact}. Therefore, alternative localization techniques have attracted considerable research focus over the past decade.

Several ground anchor based localization techniques have been extensively studied in the literature \cite{han2016survey}. In particular, the received signal strength (RSS) technique is attractive due to its intrinsic simplicity and the fact that the RSS estimation functionality is readily available in all chipsets \cite{Zanella}. However, the variation around the mean signal power due to shadowing significantly affects the reliability of this method. This is particularly important in urban areas where the shadowing effect is more severe and hence the localization accuracy drops remarkably. To address this issue, unmanned aerial vehicles (UAVs) deployed as aerial anchors is an emerging solution in order to localize ground devices \cite{perazzo,pinotti,hazem2,Rubina}. The main benefits of UAV anchors are their higher probability of line-of-sight (LoS) with ground terminals and less shadowing effect at higher altitudes \cite{hourani2}. Therefore, aerial anchors potentially are capable of resolving the main drawback of ground node localization when using RSS technique. In fact, UAV anchors can combine the benefits of satellites with a good link probability of LoS and the advantages of ground anchors with a short link length and hence higher received signal strength.

On the other hand, UAVs are typically battery-limited which introduces an important challenge towards their deployment as aerial anchors. This fact restricts UAVs operational life-time and hence reduces the number of measurements that can be collected during their mission, which can negatively affects the accuracy of localization. In fact, depending on the hovering duration, speed of the UAV, and length of the path, the energy consumption of the UAV varies. Moreover, when the probability of LoS is low, e.g., at low altitudes, more measurement locations should be sampled due to sever shadowing effect. However, at lower altitudes the range is shorter and the link budget hence better. The range is also shorter when the UAV is circling around a given area, giving localization services only to a small area at the ground. There clearly is a trade-off between the area that can be served, the energy consumption of the UAV and the localization accuracy. These trade-offs given realistic path loss, shadowing and energy consumption models derive our study in this paper. 

\subsection{Related Works}

The localization problem using terrestrial anchors (TAs) is well investigated in the literature \cite{han2016survey,Zanella}. Moreover, several recent works addressed the case where aerial anchors such as UAVs are used for localizing terrestrial nodes (TNs) \cite{perazzo,pinotti,hazem2,Rubina}. The path planning of a single mobile UAV for localizing TNs is addressed in \cite{pinotti} and \cite{perazzo}. In Particular, authors in \cite{pinotti} defined a bound on the positioning error of terrestrial nodes (TNs). In \cite{perazzo}, on the other hand, path planning algorithms that allow a drone to measure and verify TNs positions securely are proposed. However, in both \cite{pinotti} and \cite{perazzo} round trip time is used for distance estimation which requires special waveforms and very reliable clock. Path planning when using RSS for localization is presented in \cite{Rubina}. A hybrid of static and adaptive paths is proposed to minimize the localization accuracy and the trajectory length. Nevertheless, in this work the RSS-distance relation is represented by a simplified model that does not consider the dependency of path loss and shadowing characteristics on the UAV altitude.

To the best of our knowledge \cite{hazem2} is the only report that takes into account the variation of path loss exponent and shadowing with UAV elevation angle in urban areas. In this study, the optimum positioning of multiple UAVs hovering at the same altitude is investigated. It has been shown that the UAVs altitude has a significant influence on the localization accuracy. In particular, UAVs at an optimum altitude provides localization accuracy remarkably higher than that obtained by ground anchors. However, this work requires multiple UAVs which is more expensive than using one UAV and requires communication coordination between them resulting in high interface complexity. Moreover, in \cite{hazem2} the UAVs energy consumption is not taken into account, which is of high importance due to their limited source of energy resulting in limited number of measurements; hence limited localization accuracy.

\subsection{Contribution and Paper Structure}

In this paper we investigate the localization of stationary TNs using a mobile aerial anchor. To this end, a rotary-wing UAV is deployed that collects RSS measurements at different waypoints. We propose a generic analytical framework that includes height-dependent path loss and shadowing effects for urban environments. Moreover, a model for UAV energy consumption which enables us to provide practical insights into the design of mobile UAV anchors is detailed. We formulate the optimization problem and study the impact of different design factors such as UAV's altitude, number of waypoints and hovering time. Furthermore, we characterize the localization coverage using the Cram\'{e}r-Rao lower bound (CRLB). Our main objective is to minimize the localization error for a given energy constraint via optimizing the UAV trajectory, which is a novel design framework that needs to jointly minimize the localization error and yet consider the on-board energy limits. Our results show that flying at the optimal altitude and trajectory radius that contains 4 waypoints brings the average localization error from around 140m to less than 70m for TNs uniformly distributed in a given area. Trajectories with three and four waypoints are examined in our simulations where in general the one with four waypoints shows better performance. Moreover, the results show that increasing the hovering time at waypoints can remarkably decrease the localization error at the cost of higher energy consumption. 


The rest of this paper is organized as follows. In Section \ref{system} we introduce system assumptions and the employed channel model. Section \ref{energy} discusses the UAV energy consumption model. Subsequently, we formulate the optimization problem for minimum localization error in Section \ref{optimization}. Section \ref{numerical} includes numerical results and guidelines for various design factors. Finally, the conclusion is presented in Section \ref{conclusion}.
\section{System Model} \label{system}
In this section we present the assumptions that our work is based on. Subsequently, the channel model is thoroughly explained.
\subsection{System Assumptions}
Consider a UAV flying at altitude $h$  and acting as a mobile aerial anchor to localize TNs. The on-board communication technology depends on specific application communication requirements. Appropriate technologies could be as advanced as LTE or WiFi, or as simple as LPWAN. We assume that the UAV follows a certain path that consists of a sequence of $waypoints$, $\boldsymbol{w} = \{ w_o, w_1, w_2, ..., w_m, w_o\}$, as illustrated in Fig. \ref{model}. Each path is characterized by its radius $R$ and the number of waypoints denoted as $M=m+1$. The direct distance between a UAV at the $i$-th waypoint and a TN to be localized is denoted by $d_i$. Moreover, as shown in the figure, $r_i$ and $\theta_i$ represent the horizontal distance and the elevation angle with regard to TN, respectively.

The first waypoint, namely $w_o$, is the home waypoint at which the UAV starts and ends the path. At each waypoint the UAV collects data from TNs from which RSS measurements are determined. Once the UAV has collected RSS measurements from a TN at three waypoints, it can determine the distances $d_i$, and consecutively the position of such a node by multilateration. Subsequently, the estimated location of the TN is updated by taking into account all the waypoints during the mission. A well-known natural representation of the RSS-distance relation is obtained from the path loss model equation.

\begin{figure}[t]
	\centering	\includegraphics[width=0.45\textwidth]{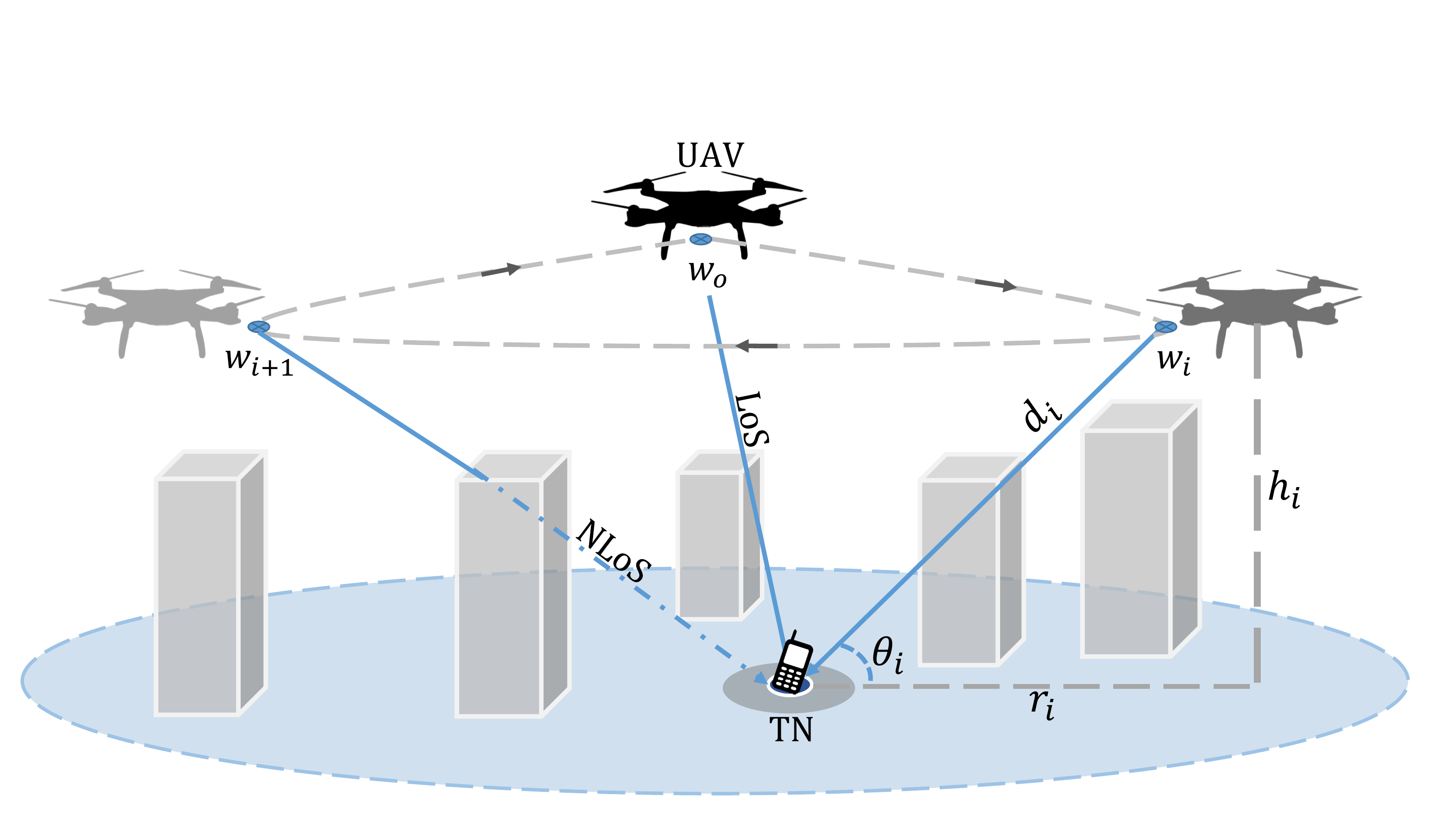} 
	\caption{An illustration of a UAV moving in a steady trajectory with waypoints $\boldsymbol{w} = \{ w_o, w_i, w_{i+1}, w_o\}$ and trying to localize TN.}
	\label{model}
\end{figure}

\subsection{Channel Model}

The communication channels between the UAV and the TNs are mainly LoS and non-LoS (NLoS). Following \cite{hourani}, we consider LoS and NLoS links separately along with their probabilities of occurrence. Accordingly, the path loss model for LoS and NLoS links in dB are respectively \cite{hourani}
\begin{subequations}\label{plmodel22}
	\begin{align}
	\text{PL}_{\text{LoS}} =20\log(d) + 20\log \left(\frac{4\pi f}{c}\right) + \psi_{\text{LoS}}, \label{plmodel1} \\ \label{plmodel2}
	\text{PL}_{\text{NLoS}} =20\log(d) + 20\log \left(\frac{4\pi f}{c}\right) + \psi_{\text{NLoS}},
	\end{align}
\end{subequations}
where $f$ is the carrier frequency, $c$ denotes the speed of light and $d$ is the distance between the UAV and the TN, given by $d = \sqrt{h^2 + r^2}$. Moreover, $\psi_{\text{LoS}}$ and $\psi_{\text{NLoS}}$ represent the variations around the mean\footnote{In \cite{hourani,hourani2} the authors called it excessive path loss. In this work the terms shadowing and excessive path loss are used interchangeably.} for LoS and NLoS, respectively. Both $\psi_{\text{LoS}}$ and $\psi_{\text{NLoS}}$ represent log-normal random variables \cite{hourani2}
\begin{eqnarray}
\psi_j \sim \mathcal{N}(\mu_j,\,\sigma_j^{2}(\theta)) ,~~~j \in \{\mathrm{LoS}\,,\mathrm{NLoS}\}
\end{eqnarray}
where $\mu_j$ is the mean and $\sigma_j^{2}(\theta)$ is the variance (in dB) given by
\begin{eqnarray}
\sigma_j(\theta) = a_j \exp{(-b_j\theta)},~~~j \in \{\mathrm{LoS}\,,\mathrm{NLoS}\}
\label{sigj}
\end{eqnarray}
with $a_j$ and $b_j$ being frequency and environment dependent parameters. The probability of having a LoS link between the UAV and TN can be expressed as
\begin{eqnarray}
\plos = \frac{1}{1 + a_o \exp{(-b_o \theta)}}, 
\label{Prlos}
\end{eqnarray}
where $a_o$ and $b_o$ are environment dependent constants and $\theta$ is the elevation angle shown in Fig. \ref{model}. Furthermore, the probability of NLoS is simply $\mathcal{P}_\text{NLoS}$ = $1 - \plos$.

\section{Energy Consumption of UAV Missions} \label{energy}

The total energy consumption of a UAV includes two components: the communication-related energy and the propulsion energy, which is required to ensure that the UAV remains aloft as well as for supporting its mobility. Practically, the communication-related energy is usually ignorable compared with the UAV’s propulsion energy, e.g., a few watts versus hundreds of watts \cite{Zeng}, and thus is ignored in this paper. In this work we consider UAV paths at which UAVs have two operational modes: the \textit{hovering mode} in which the UAV collects data from the TN and the \textit{forward flight} mode in which the UAV moves from one waypoint to another, in an straight line connecting the two waypoints.

In fact for both flight modes there are three main sources of power consumption \cite{Zeng,flight}:
\begin{itemize}
\item \textit{Blade Profile} is the power required just to turn the rotors' blade.
\item \textit{Parasitic power} which is the power used to overcome the drag force that results when UAV moves through air. This parasitic power is proportional to the cube of the UAV airspeed $v$, making it zero when hovering and very large at high speeds.
\item \textit{Induced power} is the power required to overcome the induced drag of lift creation, which is an aerodynamic drag force that occurs whenever a moving object redirects the airflow coming at it. The induced power is inversely proportional to the UAV airspeed. When hovering, all the airflow which is available for lift creation must be generated by the rotation of the main rotors. This means that a small amount of air must be considerably accelerated. However, when the UAV adds forward speed, it can achieve a higher mass flow through the rotor i.e., in forward flight the blade disc is getting more lift due to the increased downward airflow. Subsequently, less acceleration of air is needed from the motors to achieve the same thrust for lift.
\end{itemize}

\begin{figure}[t]
	\centering
	\input{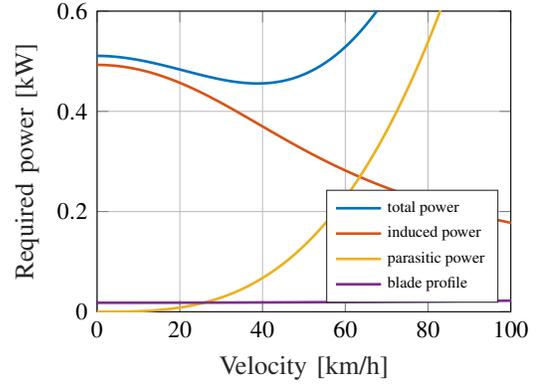}
	\caption{\small{Propulsion powers for a 5kg UAV versus the horizontal velocity}}
	\label{indPar}
\end{figure}

Now assuming a small tilt angle\footnote{The angle that represents the small tilt UAV does in forward flight mode.} ($<5^\circ$) of the UAV at forward flight mode, by following \cite{flight} and \cite{fflight}, the total power consumption can be written as:

\begin{eqnarray}
P_{ff} = \underbrace{mg v_{ind}}_{\text{induced}} + \underbrace{\frac{1}{2}\rho v^3 C_{ds}}_{\text{parasitic}} + \underbrace{k_o\left( 1+ 3\frac{v^2}{v_t^2} \right)}_{\text{blade profile}},
\label{tpower}
\end{eqnarray}
where $C_{ds}$ is a constant depends on the UAV drag coefficient and the reference area, $\rho$ is known as the air density, $g$ denotes the standard gravity, $m$ is the mass of the UAV in kilograms, $k_o$ is a constant depends on the dimensions of the blade and $v_t$ is the tip speed of the rotor blade. In (\ref{tpower}), moreover, $v_{ind}$ is the mean propellers' induced velocity in the forward fight mode, given by
\begin{eqnarray}
v_{ind} = \sqrt{\frac{-v^2 + \sqrt{v^4 + \left(\frac{mg}{\rho A_d}\right)^2 }}{2}},
\label{indV}
\end{eqnarray}
where $A_d$ is the area of the UAV. In case of hovering, i.e., \mbox{$v$ = 0}, the power required for a quad-copter UAV to hover is
\begin{eqnarray}
P_h = k_o + \sqrt{\frac{(mg)^3}{2\rho A_d }}.
\label{hEngy}
\end{eqnarray}

In Fig. \ref{indPar} we show the typical trends of the three propulsion powers individually along with the total power consumption versus UAV velocity $v$. As shown in the figure, when \mbox{$v$ = 0} the UAV consumes more power compared with forward flight at optimal velocity. Therefore, one has to minimize the number of waypoints at which the UAV hovers with $v$ = 0 and fly with the optimal forward velocity to minimize the energy consumption.

It is worth noting that the power consumption model introduced in this section is not a new aerodynamic model for the power consumption of rotary-wing UAVs. However, it is a simplified model that is suitable for researchers in communications theory. Interested readers may refer to \cite{flight} for more comprehensive theoretical derivations.

\section{Trajectory Design for Minimizing The Localization Error} \label{optimization}

Here, we formulate the optimization problem under consideration. First we introduce the localization error model and subsequently, present the coverage of the UAV during the mission. Second, we propose our framework of choosing the design parameters of the trajectory in order to minimize the localization error. Trajectory design parameters are altitude, hovering time, number of waypoints and the straight distance between them. 

\subsection{Localization Error and Coverage}

Following the channel model presented in Section \ref{system} and assuming that $\mathcal{E}_{\text{LoS}}$ and $\mathcal{E}_{\text{NLoS}}$ respectively represent the localization errors corresponding to LoS and NLoS components, the average localization error can be written as
\begin{eqnarray}
\mathcal{E} = \plos(\theta)~\mathcal{E}_{\text{LoS}} + [1 - \plos(\theta)] \mathcal{E}_{\text{NLoS}}.
\label{Errtotal}
\end{eqnarray}
Without loss of generality we assume a 3D Cartesian coordinate in which the TN location is $(x_g, y_g, 0)$ whereas the UAV position at any given time is $(x_a, y_a, h)$. Consequently, given the estimated distance $\hat{r}_i$ and known projection $(x_a^{(i)}\,, y_a^{(i)})$ of the UAV at the $i$th waypoint, the position of the TN can be estimated by finding the point $(\hat{x}\,, \hat{y})$ that satisfies
\begin{eqnarray}
(\hat{x}, \hat{y}) = \argmin_{x,y} \bigg\{\sum_{i=1}^{M} \Big( \sqrt{(x_a^{(i)} - x_g)^2 + (y_a^{(i)} - y_g)^2} - {\hat{r}}_i  \Big)^2 \bigg\}.
\label{opt2}
\end{eqnarray}
where $\hat{r_i} = \sqrt{\hat{d}_i - h^2}$ and $\hat{d}_i$ is the estimated distance obtained from (\ref{prx}) for LoS or NLoS link. Now, for an estimated location $(\hat{x}\,, \hat{y})$ of a TN, the localization error is expressed as
\begin{eqnarray}
\mathcal{E}_j = \hspace{0.1cm} \parallel {\hat{\boldsymbol{r}}} - \boldsymbol{r} \parallel \hspace{0.1cm} =  \sqrt{\sum_{i=1}^{M} \mid\hat{r}_i - r_i \mid^2}, ~~~j \in \{\mathrm{LoS}\,,\mathrm{NLoS}\}
\label{locoErr}
\end{eqnarray}
where $\boldsymbol{r} = [r_1, r_2, ..., r_M]$, ${\hat{\boldsymbol{r}}} = [\hat{r}_1, \hat{r}_2, ..., \hat{r}_M]$ and $\lVert.\lVert$ represents the euclidean distance.

In order to localize a given TN and to exploit the benefits of the trajectory, it must present in the coverage region of the UAV at all waypoints. Mathematically speaking, the UAV coverage for the distance estimator bounded by CRLB at the $i$-th waypoint is written as
\begin{eqnarray}
r_c^{(i)} = r|_{\sigma(\hat{d}) = \delta \sigma_{CRLB}},
\label{locCov}
\end{eqnarray}
where $\sigma(\hat{d})$ is the standard deviation of the distance estimator, $\delta$ is the localization coverage factor and $\sigma_{CRLB}$ denotes the CRLB. The CRLB is known as one of the most important performance benchmarks for ranging estimators \cite{van2004detection}. Based on the channel model considered in this work, the average CRLB is expressed as
\begin{eqnarray}
\sigma^2_{\text{CRLB}}(\theta) &=& \plos^2(\theta)~\sigma^2_{\text{CRLB}\,|\,\text{LoS}}(\theta) \nonumber\\ 
&+& [1 - \plos(\theta)]^2 \sigma^2_{\text{CRLB}\,|\,\text{NLoS}}(\theta),
\label{crlbTotal}
\end{eqnarray}
where $\sigma_{\text{CRLB}\,|\,\text{LoS}}$ and $\sigma_{\text{CRLB}\,|\,\text{NLoS}}$ are defined in the following theorem.
\begin{theorem}
	The Cram\'{e}r-Rao lower bound for the RSS-based distance estimator using an aerial anchor is given by:
	\begin{eqnarray}
	\sigma_{\text{CRLB}\,|\,\text{LoS}} &\geq&  \frac{d \, \ln (10)}{20} \,\,a_{\text{LoS}}\exp(-\theta b_{\text{LoS}})
	\label{TcrlFinallos}
	\end{eqnarray}
	\begin{eqnarray}
	\sigma_{\text{CRLB}\,|\,\text{NLoS}} &\geq&  \frac{d \, \ln (10)}{20} \,\,a_{\text{NLoS}}\exp(-\theta b_{\text{NLoS}})
	\label{TcrlFinalnlos}
	\end{eqnarray}
	for LoS and NLoS, respectively.
\end{theorem}

\begin{proof}
	The proof is given in Appendix A
\end{proof}

\begin{figure}[t]
	\input{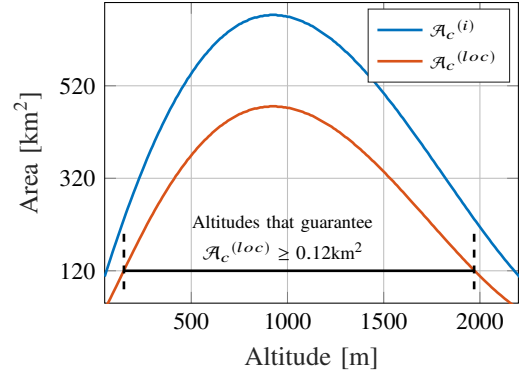}
	\caption{\small{Localization coverage versus UAV altitude. $M = 3$ and $R = 120$m.}}
\label{LocCover}
\end{figure}

Now the radius $r_c^{(i)}$ in (\ref{locCov}) represents a coverage area ${\mathcal{A}_c}^{(i)}$. Consequently, one can define the \textit{localization coverage area} as the intersection of the UAV coverage areas at all trajectory waypoints in which the distance estimation error is bounded by a factor of the Cram\'{e}r-Rao lower bound. Mathematically, for a trajectory with $M$ waypoints it is given by
\begin{eqnarray}
{\mathcal{A}_c}^{(loc)} = \bigcap_{i=1}^M {\mathcal{A}_c}^{(i)}.
\label{totCov}
\end{eqnarray}
For a trajectory of three waypoints with identical coverage radius, i.e., $r_c^{(1)} = r_c^{(2)} = r_c^{(3)} = r_c$ , a closed-form expression for the intersection area is given by \cite{fewell2006area}
\begin{multline}
{\mathcal{A}_c}^{(loc)} = \frac{\sqrt{3}}{4}c^2 + 3 \left( r_c^2 \arcsin \left(\frac{c}{2r_c}\right) - \frac{c}{4}\sqrt{4r_c^2 - c^2} \right) ,
\label{intArea}
\end{multline}
where 
\begin{eqnarray}
c = \sqrt{ 3r_c^2 - \frac{l^2}{2} - l\sqrt{3r_c^2 - \frac{3l^2}{4}} }
\end{eqnarray}
and $l$ is the distance between any two adjacent waypoints. As an example, Fig. \ref{LocCover} presents ${\mathcal{A}_c}^{(loc)}$ together with ${\mathcal{A}_c}^{(i)}$ for different UAV altitudes. The figure shows that for $h \in [150, 1900]$ one guarantees a minimum localization coverage of 0.12km$^2$.  

\subsection{Trajectory Design}

Consider a trajectory of radius $R$ and $M$ waypoints as shown in Fig. \ref{wpoints}, the distance between any two adjacent waypoints in the trajectory is given by
\begin{eqnarray}
l_{M} = 2 R \sin \left( \frac{\vartheta}{2} \right),
\label{trjLen}
\end{eqnarray}
where $\vartheta = \frac{2\pi}{M}$ is the angle between any two adjacent waypoints within the trajectory. Using (\ref{trjLen}) one can control the trajectory design parameters $M$ and $R$. For instance, in Fig. \ref{wpoints} we fixed $R$ and plotted the cases for $M=3$ and $M = 4$. Once $M$ and $R$ are defined, the energy consumption of the trajectory can be written as  
\begin{eqnarray}
{E}_t = \sum_{i = 1}^{M} t_h^{(i)}\,P_h^{(i)} + \sum_{j = 1}^{M} \frac{l_M^{(j)}}{v} P_{ff}^{(j)},
\end{eqnarray}
where $t_h^{(i)}$ is the hovering time at the $i$-th waypoint. Now assuming a limited on-board energy of $E_{th}$, in order to find the design parameters that minimize the localization error, the optimization problem can be written as
\begin{equation}
\begin{aligned}
& \underset{h, \,R, \,M, \,t_h}{\text{minimize}}
& & {\mathcal{E}} \\
& \text{subject to}
& & \tilde{\text{N}} = N, \\
&&& \sum_{i = 1}^{M} t_h^{(i)}\,P_h^{(i)} + \sum_{j = 1}^{M} \frac{l_M^{(j)}}{v} P_{ff}^{(j)} = E_{th}, \\
&&& M \geq 3, \\
\label{optErr2}
\end{aligned}
\end{equation}
where $N$ is the total number of TNs and $\tilde{\text{N}}$ is the number of TNs within the localization coverage. Based on the formulation provided in (\ref{optErr2}), various design parameters influence the localization error. Firstly, $h$, as changing $h$ affects both the shadowing and the path loss. From (\ref{sigj}), increasing $h$ decreases the variation around the mean which positively affects the localization accuracy. However, concurrently, the slop of the curve decreases due to the logarithmic RSS-distance relation. At low slop, small variations produce large localization error. Secondly, the radius, at which the trajectory is design around, is crucial for both the localization accuracy and the energy requirements. In one hand, large values of $R$ are required, preferably bigger than the radius of the serving area, so that the multilateration method can preform better \cite{perazzo}. On the other hand large values of $R$ implies longer trajectories and hence more energy consumption. Another important parameters in (\ref{optErr2}) are the hovering time and the number of waypoints. Intuitively, the more RSS measurements the better the localization accuracy. More RSS measurements can be collected by either adding more waypoints or increasing the hovering time at each. However, hovering with strictly zero speed is known to be energy-inefficient for rotary-wing UAVs. Thus, the energy-constrained trajectory design on $E_{th}$ needs to strike an optimal balance between maximizing the localization accuracy and minimizing the UAV’s propulsion energy consumption.

\begin{figure}[t]
	\centering	\includegraphics[width=0.5\textwidth]{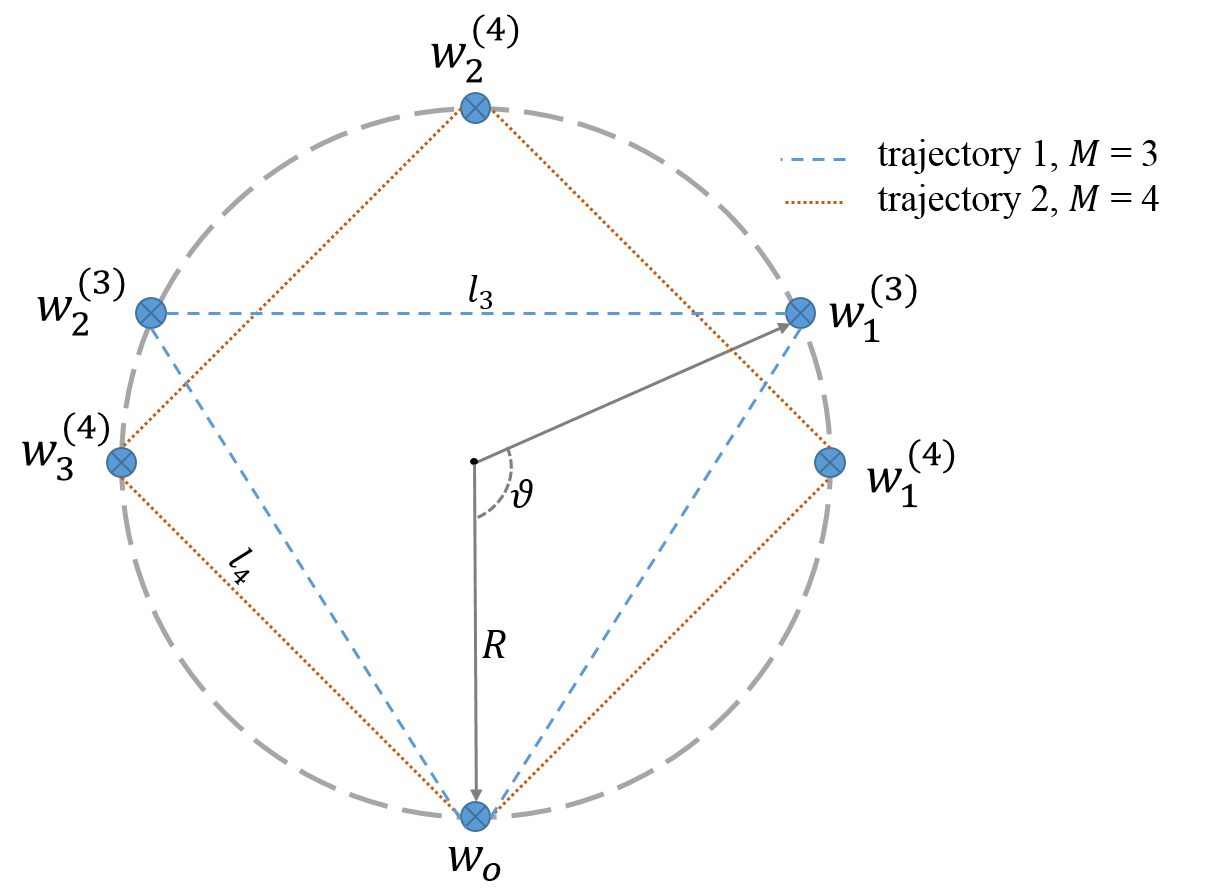} 
	\caption{An elevation view of UAV trajectories with 3 and 4 waypoints}
	\label{wpoints}
\end{figure}

\begin{table}[t]
	\caption{\small{Parameters list}}
	\vspace{-1.5em}
	\begin{center}
		\begin{tabular}{ | l || c | c |}
			\hline
			\textbf{Parameter} & \textbf{Description} &  value \\
			\hline
			\hline
			$M$ & Number of waypoints &  3 , 4 \\
			\hline
			$N$ &  Number of TNs &  100 \\
			\hline
			$f$ & Carrier frequency [GHz]&  2  \\
			\hline
			$A$ &  Total area of TNs [km$^2$] & 0.12 \\
			\hline
			$\mathcal{E}$ &  Localization error & -- \\
			\hline
			$E_{th}$ & Energy threshold & -- \\
			\hline
			$l_j$ &  Waypoints inter distance [m] & -- \\ 
			\hline
			$h$ & UAV's altitude & 200 \\
			\hline
			$h$ & Hovering time [s] & 5 \\
			\hline
			$R$ &  Trajectory radius & 120 \\
			\hline
			$\rho$ &  Air density & 1.225 \\ 
			\hline
			$m$ &  UAV mass [kg] & 5 \\ 
			\hline
			$C_{ds}$ & drag and reference area coefficient  & 0.4 \\ 
			\hline
			$v_t$ &  Tip speed of the blade & 100 \\ 
			\hline
			$A_d$ &  UAV's surface area [m$^2$] & 0.25  \\ 
			\hline
			$v$ &  UAV's forward velocity [km$/$h] & 40 \\ 
			\hline
			$k_o$ &  Blade dimension constant & 570 \\ 
			\hline
			$a_{\text{LoS}}$ & Shadowing constant  & 10\\
			\hline
			$b_{\text{LoS}}$ &  Shadowing constant & 2 \\
			\hline
			$a_{\text{NLoS}}$ & Shadowing constant & 30\\
			\hline
			$b_{\text{NLoS}}$ &  Shadowing constant & 1.7 \\
			\hline
			$a_o$ & $\plos$ constant  & 47 \\
			\hline
			$b_o$ &  $\plos$ constant & 20 \\
			\hline
			$\delta$ &  Localization coverage factor & 2 \\
			\hline
			
		\end{tabular}
	\end{center}
	\label{const}
	\vspace{-1em}
\end{table}

\section{Case Study: Trajectories with 3 and 4 Waypoints} \label{numerical}

In this section we investigate UAV trajectory design for localization, numerically. In our simulations we assume 100 TNs uniformly distributed in a circular area with a radius of 200m, centered at $(x,\, y) = (0,\,0)$. We consider a system communication frequency of 2GHz. Moreover, we assume that, the hovering time $t_h$ is equal at all waypoints. During the hovering time $t_h$ at any given waypoint, the UAV collects data from TNs from which RSS is measured in rate of 2 RSS samples per second. The relevant system parameters and their corresponding values are specified in details in Table \ref{const}, unless mentioned otherwise. In the following we investigate each design parameter of the proposed framework.

\begin{figure}[t]
	\centering
	\input{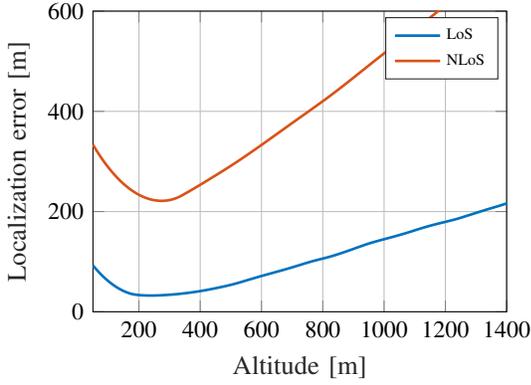}
	\caption{\small{Localization error for LoS and NLoS versus $h$. $R$ = 120m, $t_h = 5$ s and $M = 3$.}}
	\label{Locerr}
\end{figure}

\begin{figure}[t]
	\centering
	\input{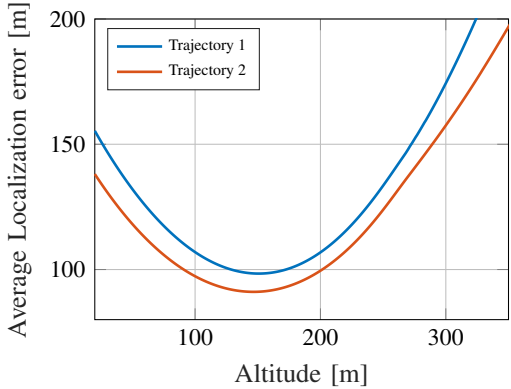}
	\caption{\small{Localization error for trajectory 1 and trajectory 2 versus $h$. $R$ = 120m, $t_h = 5$ s.}}
	\label{trajAlt}
\end{figure}

\begin{figure}[t]
	\centering
	\input{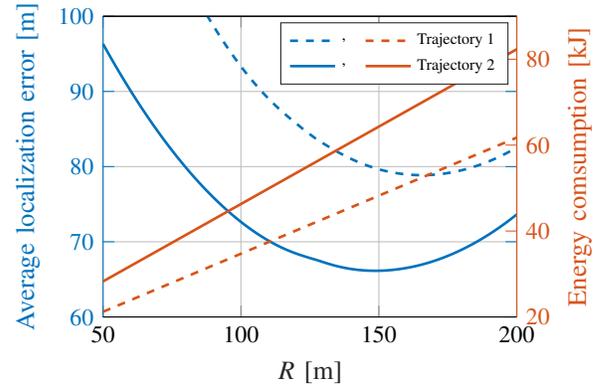}
	\caption{\small{Localization error of trajectory 1 with 3 waypoints and trajectory 2 with 4 waypoints. $h$ = 200m and $t_h = 5$ s.}}
	\label{trajloco}
\end{figure}

\begin{figure}[t]
	\centering
	\input{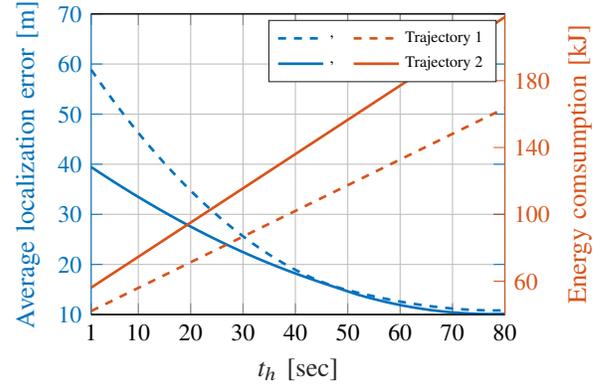}
	\caption{\small{Hovering time at each waypoint vs localization error for a TN at the origin. $R= 120$m, $h = 200$m.}}
	\label{ptotime}
\end{figure}

\subsubsection{UAV Altitude}
Evidently, Fig. \ref{Locerr} and \ref{trajAlt}, show that the UAV anchor outperforms the ground anchor 
when flying at the optimal altitude. In Fig. \ref{Locerr}, the localization error assuming LoS and NLoS for a trajectory with $M = 3$ is illustrated. As one can see, for both LoS and NLoS cases the error decreases and then increases with $h$. The error decreases with $h$ because of the exponentially decreasing variance of the $\psi_j$ with $h$ \cite{hourani2}. On the other hand, for large values of $h$, and hence $d$, the low resolution curve will be more vulnerable to excessive path loss effects (i.e., tiny variations in the path loss model curve will lead to a large estimation error), making localization accuracy inversely related to $h$. Finally the figure also shows that The localization error is always better for a LoS channel.

The localization error follows the same trend when it is averaged over LoS and NLoS using (\ref{Errtotal}) as shown in Fig. \ref{trajAlt}. The figure also compares the localization error of trajectories 1 with 3 waypoints, and 2 with 4 waypoints, shown in Fig. \ref{wpoints}. Interestingly, we can see that optimizing the trajectory's altitude can decrease the localization error from 150m to around 90m. Moreover, the trajectory with 4 waypoints provides better performance than the one with 3 waypoints for the same hovering time at each waypoint. 

\subsubsection{Trajectory Radius}
In Fig. \ref{trajloco} we present how the trajectory radius influences the localization error performance when $h$ and $t_h$ are fixed. Furthermore, the figure presents the energy requirements for each radius. The figure shows that high localization errors occur when $R$ is small (i.e., $R = 50$). This is due to the fact that, for multilateration, at small distances between the anchor points, a small estimation error in the distance will lead to a large error in the estimated location. Hence, increasing $R$ decreases the localization error from 150m to 80m in trajectory 1 and from 95m to 65m in trajectory 2 at optimal $R$. The cost here is the energy required for larger $R$. For instance, 55kJ and 70 kJ are required at optimal $R$ for trajectory 1 and trajectory 2, respectively.

\subsubsection{Hovering Time}
Fig. \ref{ptotime} shows the localization error as a function of hovering time. Increasing the hovering time implies more RSS measurements collected which is very beneficial to improve the accuracy at the cost of increased energy consumption. As one can see in the figure, after 40 seconds the two trajectories approach the same behavior based on our simulations. An interesting finding here, is that optimizing the trajectory matters more in the cases of low hovering time and hence, limited energy. Otherwise, a trajectory with three waypoints would be enough if the energy on-board can supports longer hovering.

\subsubsection{Number of Waypoints}

The number of waypoints is yet another important parameter in trajectory planning for localization of TNs. In general, adding more waypoints improves the system performance as shown in Fig. \ref{trajAlt}, Fig. \ref{trajloco} and Fig. \ref{ptotime}. Nevertheless, injecting more waypoints in the trajectory implies a longer total hovering time and longer traveling distances for the UAV, leading to a higher energy consumption.

\section{Conclusion} \label{conclusion}
In this paper, we studied the use of a moving UAV for ground node localization in urban environments. We proposed a generic analytical framework that enables us to study various trade-offs when designing UAV networks for localization. We provided a qualitative and quantitative understanding of how the design parameters such as UAV altitude, hovering time, traveling distance and number of waypoints affect the localization service while considering the on-board energy-constrained. Our findings show that mobile UAV anchor is capable of providing a desirable localization service when design parameters are optimized, outperforming that obtained by ground anchors. In particular we showed that the UAV altitude has a major influence on the localization accuracy compared to other design parameters. However, for a short hovering time the trajectory optimization is more crucial.



%

\appendices
\section{Proof of  Theorem 1}
The time-averaged received power (in dBm) for $j$ $\in$ \{LoS, NLoS\} can be written as
\begin{eqnarray}
\mathcal{P}_{r}^{(j)} = - 20\log(d) - K - \, \psi_{{j}} + C,
\label{prx}
\end{eqnarray}
where C is a constant (in dBm) which depends on the transmit power and received power to RSS transduction and $K = 20\log \left(\frac{4\pi f}{c}\right)$. Note that for a given distance $d$, the time-averaged received power given in (\ref{prx}) is a stochastic variable following a shifted version of the probability density function (PDF) of $\psi(\theta)$. Accordingly, the PDF of $\mathcal{P}_{r}$ conditioned on $d$ and $\theta$ is given by
\begin{eqnarray}
f_{\mathcal{P}_{r}|d,\theta}^{({j})}(w) = f_{\psi | d,\theta}( -w - 20\log(d) - K + C),
\label{Pr-pdf}
\end{eqnarray}
where $w$ is an auxiliary variable. The CRLB of the estimated distance denoted as $\hat{d}$ is then expressed as \cite{van2004detection}
\begin{eqnarray}
\sigma_{\text{CRLB}}^2 \geq \frac{1}{E_{\mathcal{P}_{r}|d,\theta}\bigg\{ \bigg[ \frac{\partial}{\partial d} \ln f_{\mathcal{P}_{r}|d,\theta}(w) \bigg]^2 \bigg\}},
\label{cr}
\end{eqnarray}
where $E_{\mathcal{P}_{r}|d,\theta}\{.\}$ is the expectation conditioned on $d$ and $\theta$. Using the PDF given in (\ref{Pr-pdf}), it is straightforward to show that
\begin{eqnarray}
\frac{\partial}{\partial d} \ln f_{\mathcal{P}_{r}|d,\theta}^{({j})}(w) &=& [-w - 20\log(d) - K + C] \nonumber\\
&\times& \frac{20}{d~\ln(10)~\sigma^2_{{j}}(\theta)}
\label{crlbDiff}
\end{eqnarray}
We now substitute (\ref{sigj}) in (\ref{crlbDiff}) and subsequently (\ref{crlbDiff}) in (\ref{cr}). Then, after simplifications, one obtains (\ref{TcrlFinallos}) and (\ref{TcrlFinalnlos}) for LoS and NLoS respectively.


\section*{Acknowledgment}

This work was funded by the SESAR project, PercEvite.

\bibliographystyle{ieeetr} 
\bibliography{IEEEbib}

\end{document}